\documentclass[11pt,letterpaper]{article}
\usepackage{microtype}
\usepackage[rflt]{floatflt}
\usepackage{hyperref}
\usepackage{fullpage}
\usepackage{graphicx}
\usepackage{epsfig}
\usepackage{color}
\usepackage{enumerate}
\usepackage{verbatim}
\usepackage{amsthm, amssymb}
\usepackage{amsmath}
\usepackage{wrapfig}

\usepackage{algorithm}
\usepackage{algorithmicx}
\usepackage[noend]{algpseudocode}
\usepackage{complexity}
\usepackage{microtype}

\algblock{Begin}{End}

\date{}
\MakeRobust{\overrightarrow}

\newtheorem{theorem}{Theorem}
\newtheorem{lemma}[theorem]{Lemma}

\newtheorem{definition}[theorem]{Definition}
\newtheorem{corollary}[theorem]{Corollary}

\newcommand{\out}[1]{}

\newcommand{\cB}{\mathcal{B}}
\newcommand{\cC}{\mathcal{C}}

 \title{\bf NC Algorithms for Weighted Planar Perfect Matching and Related Problems\thanks{I would like to thank anonymous reviewers for helpful comments about
previous version of this paper.}}
\author{
Piotr Sankowski\\Institute of Informatics\\
University of Warsaw\\ {\tt sank@mimuw.edu.pl}}

\begin{document}

\maketitle

\begin{abstract}
Consider a planar graph $G=(V,E)$ with polynomially bounded edge weight function
$w:E\to [0, \poly(n)]$. The main results of this paper are NC algorithms for the following problems:
\begin{itemize}
\item minimum weight perfect matching in $G$,
\item maximum cardinality and maximum weight matching in $G$ when $G$ is bipartite,
\item maximum multiple-source multiple-sink flow in $G$ where $c:E\to [1, \poly(n)]$ is a polynomially bounded edge capacity function,
\item minimum weight $f$-factor in $G$ where $f:V\to [1, \poly(n)]$,
\item min-cost flow in $G$ where $c:E\to [1, \poly(n)]$ is a polynomially bounded edge capacity function
and $b:V\to [1, \poly(n)]$ is a polynomially bounded vertex demand function.
\end{itemize}
There have been no known NC algorithms for any of these problems previously.\footnote{Before this and independent paper by Anari and Vazirani~\cite{1709.07822}.} In order
to solve these problems we develop a new relatively simple but versatile framework that is combinatorial
in spirit. It handles the combinatorial structure of matchings directly and needs to
only know weights of appropriately defined matchings from algebraic subroutines.
\end{abstract}

\section{Introduction}
In this paper we study deterministic parallel algorithms for the maximum planar matching problem. In particular, we concentrate our attention
on the NC class where we are given polynomially many processors which need to solve the problem in polylogarithmic time. The fundamental
algorithmic challenge is to compute the maximum cardinality matching, or maximum weight matching, or minimum weight perfect matching.

So far in the case of NC algorithms for planar graphs there seemed to be no tools to attack the most general case and we were only
able to find perfect matchings only in planar bipartite graphs~\cite{MILLER86,Mahajan:2000}. This might be bit surprising as in non-planar graph,
perfect matchings are essentially as powerful
as maximum cardinality matchings, see e.g.,~\cite{Sankowski:2007:FDM:1283383.1283397} for a reduction. Such reductions either make two copies of the graph, where
each vertex is connected with its copy, or add a set of vertices that are connected to every vertex in the graph. There exists no similar planarity preserving reduction,
and so usefulness of algorithms for planar perfect matchings is limited. Hence, in order to find a maximum cardinality matching we need to
retool to less efficient polynomial time algorithms~\cite{Gabow:1988:ASA:62212.62263,21935,Goldberg:1989:IMP:1398514.1398715} that are limited to bipartite case only.
Or, alternatively, we can find a $2$-approximate solution~\cite{Hanckowiak:1998:DCC:314613.314705}. The lack of such planarity preserving reduction is
one of the reasons why computing maximum cardinality matchings in almost linear time was recognized as one of the
core problems in planar graph algorithms~\cite{6108162}. This despite the fact that almost linear
time algorithms for the perfect matching problem existed before~\cite{miller-95,FAKCHAROENPHOL2006868}.

Formally speaking, our results are as follows. Consider a planar graph $G=(V,E)$ with polynomially bounded edge weight (cost) function
$w:E\to [0, \poly(n)]$. A {\em matching} in $G$ is an independent set of edges, whereas a {\em perfect matching} is matching
that is incident to every vertex in $G$. We start by giving a relatively simple NC algorithm for finding a perfect matching in $G$.
Then we extend this algorithm to finding {\em minimum perfect matching}, i.e.,
a perfect matching $M$ of $G$ with minimum total weight $w(M)$. Our algorithm is combinatorial in spirit as it does not manipulate
a fractional matching, but only requires to know weights of appropriately defined minimum perfect matchings. These weights can be computed
using Kasteleyn's Pfaffian orientation of planar graphs and standard algebraic techniques. The algorithm is based
on the fundamental notion of balanced duals that was developed in~\cite{CGSa} as well as~\cite{Gabow:2012}. This idea allows
to define and construct a dual solution in a unique way. In particular, a very simple NC algorithm (see Appendix~\ref{algorithm-blossoms}) for constructing a family of blossoms
of this dual follows by a direct application of the algorithm given in~\cite{CGSa}.
\footnote{Our framework
was developed independently from~\cite{1709.07822} and posted on arXiv on the same day~\cite{1709.07869}. Nevertheless, the author's initial write-up
was very "crude" and contained some gaps. In particular, Algorithm~1 was described in a confusing way without stating which parts should happen in parallel.
Moreover, we aimed to give, an alternative to~\cite{CGSa}, construction
of matching duals that was incorrect. This, however, is a known result, so in this paper we just cite~\cite{CGSa} instead. Finally, the
presentation of this paper was greatly improved with respect to~\cite{1709.07869}, while keeping the original framework. Moreover,
consequences for more general problems have been added. }

The next problem we consider is the {\em minimum weight $f$-factor problem}, where for a given $f:V\to [1, \poly(n)]$
we need to find a set of edges $F$ of minimum total cost $w(F)$ such that $\deg_F(v) =f(v)$ for all $v\in V$. Typically,
this problem is reduced to the perfect matching problem via vertex splitting~\cite{hal83,anstee}. These reductions, however, do not
preserve planarity. Our contribution, is to show a new planarity preserving reduction, that allows to solve the $f$-factor problem
in NC.

The following implication of this result was rather surprising (at least for the author). We show that both the maximum cardinality
and maximum weight bipartite planar matching
problem can be efficiently reduced to the minimum non-bipartite planar matching problem. Thus implying the first known NC algorithm
for finding maximum weight matching and maximum cardinality matching in bipartite planar graphs. We note that prior to this result there have
been no tools or ideas indicating that this problem could have an solution in NC. One reason for this is the applicability of Kasteleyn's result to
the perfect matching problem only. Pfaffian orientations fail for maximum size matchings. Moreover, our reduction preserves the size of the graph, thus
any further development for weighted perfect matching in planar graphs will imply similar results for  maximum weight bipartite matchings.
There seem to be no easy way to extend this result to non-bipartite case, and we note that finding NC algorithms for this problem remains open.
Still, we report partial progress on this problem by showing the first $o(n)$ time PRAM algorithm.

Finally, we consider directed flow networks where capacity of edges is given by a polynomially bounded capacity function  $c:E\to [1, \poly(n)]$
and vertex demand is given by polynomially bounded demand function $b:V\to [1, \poly(n)]$. Our aim is to compute a min-cost flow
that obeys edge capacities and satisfies all vertex demands.
We give the first known NC algorithm for this problem resolving the long standing open problem from~\cite{miller-95}. This results is
further extended to finding maximum multiple-source multiple-sink flow in planar graphs.
\paragraph*{Related Work}
The key result that allowed development of NC algorithms for the perfect matching problems in planar graph is Kasteleyn's idea~\cite{kasteleyn}. He showed
that the problem of counting perfect matchings in planar graphs is reducible to determinant computation. Thus~\cite{csanky}
implied an NC algorithm checking whether a planar graph contains a perfect matching. In~\cite{VAZIRANI1989} this result was extended to $K_{3,3}$-free graphs.
Counting, however, does not allow to construct a perfect matching.
The first algorithm for constructing perfect matchings was given in~\cite{miller-95} but only in the bipartite case.
Another solution that is based on different principles was developed in~\cite{Mahajan:2000}. A partial solution for the non-bipartite case
was given in~\cite{Kulkarni2004} where an algorithm for computing half integral solution was shown. Finally, the problem has been solved for non-bipartite
planar graphs independently in~\cite{1709.07822} and here. The result of~\cite{1709.07822} extends to weighted case as well. In comparison, our
paper gives a simpler algorithm for constructing a perfect matchings in NC than the one in~\cite{1709.07822}. Moreover, we show several nontrivial
extension of this results, e.g., to cardinality bipartite matching.

When allowing randomization (i.e., when considering RNC complexity) the problem can be solved even in general non-planar graphs~\cite{karp-upfal-wigderson-86,galil-pan-88,mulmuley-vazirani-vazirani-87,Sankowski05p}. All of these Monte Carlo algorithms
can be changed into Las Vegas ones using~\cite{karloff-86}. Similarly, matching problems can be solved for a special graph class
which polynomially bounded number of matchings~\cite{5497892,Agrawal:2007:PBP:1763424.1763483}. However, these results do
not extend to the case of superpolynomial number of matchings or alternatively superpolynomial number of even length cycles.

Finally, we note that we are unaware of any previous parallel deterministic algorithms for weighted non-bipartite problems like minimum perfect matching or $f$-factor problems.
Bipartite versions of these problem have some solutions that require at least polynomial $\Omega(n^{2/3})$ time \cite{Gabow:1988:ASA:62212.62263,21935,Goldberg:1989:IMP:1398514.1398715}.

\section{Preliminaries}
$G = (V,E)$ denotes an $n$-vertex, embedded, undirected graph.  This embedding partitions
the plane into maximal open connected sets and we refer to the closures of these sets as the \emph{faces} of $G$. The number of faces is denoted by $f$.
For a subset of vertices $U\subseteq V$, $\delta(U)$ denotes
all edges $uv \in E$ having $|\{u, v\}\cap U|=1$. We write $\delta(u)$ for $\delta(\{u\})$,
 $x(F)$ for $\sum_{e\in F} x_e$ and $\deg_F(u)=|F\cap \delta(u)|$.

The linear programming formulation of the minimum perfect
matching problem and its dual is as follows~\cite{e}. An {\em odd set} has odd size; $\Omega$ is
the collection of odd subsets of $V$ of size $\ge 3$.
\newline

\begin{tabular}{cp{0.00\textwidth}|p{0.00\textwidth}c}
LP of minimum prefect matching &&& LP of the dual prolem\\
$
\begin{aligned}
\min \sum_{e\in E} w(e) x_e && \nonumber\\
x(\delta(v))&=&1, \forall v\in V\label{eqn:primal_v}\\
x(\delta(U)) &\ge& 1, \forall U \in \Omega \label{eqn:primal_b} \\
x_{e} &\ge& 0, \forall e \in E\label{eqn:primal_e}
\end{aligned}
$
&&&
$
\begin{aligned}
\max \sum_{v\in V} \pi_v + \sum_{U \in \Omega} \pi_U&&\\
\pi_u + \pi_v + \sum_{U\in \Omega,\ uv\in \delta(U)} \pi_U  &\le& w(uv), \forall uv\in E \label{eqn:dual}\ \ (*)\\
\pi_U &\ge& 0, \forall U \in \Omega\ \ \ \ \ \ \\
\end{aligned}
$
\end{tabular}
\newline

The variables $x_{e}$ in primal indicate when an edge is included in the solution. The dual problem has variables $\pi_v$ for each vertex $v$ and $\pi_U$ for each odd set $U$.

Moreover, a graph $G$ is {\em factor critical} if for all $v\in V$ after removing $v$ the graph has a perfect matching. A {\em laminar family}
is a collection $\cB$ of subsets of $V$ such that for every pair $X,Y \in \cB$ either $X\cap Y=\emptyset$, or $X\subseteq Y$, or $Y\subseteq X$.
We use the existence of the following dual.
\begin{lemma}[implicitly in \cite{e}]
\label{lem:critical}
There exists an optimal dual solution $\pi : V \cup \Omega \rightarrow \mathbb{R}$ that:
\begin{enumerate}
  \item the set system $\{U \in \Omega : \pi_U > 0\}$ forms a laminar family,
  \item for each $U \in \Omega$ with $\pi_U > 0$, the graph denoted by $G_U$ obtained
  by contracting each set $\{S \in \Omega : S \subset U, \pi_S>0\}$ to a point is factor critical.
\end{enumerate}
\end{lemma}
An optimum dual solution $\pi$ satisfying the above conditions is a {\em critical dual solution}.
A set $U \in \Omega$ such that $\pi_U > 0$ is a {\em blossom} w.r.t. $\pi$. An important idea that is
used in almost all algorithms for weighted matching is that after computing the dual we can work with a non-weighted problem.
This non-weighted problem is defined in the following way: leave only {\em tight} edges in the graph, i.e., there is equality in (*);
find a perfect matching {\em respecting} all blossoms $\cB$, i.e., such that for all $B\in \cB$ we have $|M\cap \delta(B)|=1$.
By duality slackness any matching obtained this way is a minimum perfect matching.

Laminar family of sets is equipped with a natural parent-child relation and can be seen as
a forrest. We assume this tree is rooted at $V$ and call the resulting tree {\em laminar tree}.
We note that given a laminar family it is straightforward to deduce the parent-child relation, i.e., parent of a set $B$ is the minimal set
containing $B$, whereas children of $B$ are maximal sets contained in $B$. Hence, whenever
working with a laminar family we assume that the laminar tree is available as well as it can be easily computed in NC.
A useful property of this view is that tree $T$  has a {\em vertex separator},
i.e., there exists a vertex $v$ such that the size of every connected component of $T-x$ is at most $\frac{|T|}{2}$.

\paragraph*{Basic NC Algorithms}
Our algorithm builds upon the following NC algorithms for computing:
\begin{itemize}
\item components and spanning forrest of an undirected graphs~\cite{connectivity-nc,CHONG1995378},
\item paths in a directed graph -- this can be done by repetitive matrix squaring,
\item maximal independent set in a graph~\cite{Luby:1985},
\item a vertex separator of a tree -- by computing the numbers of vertices in each subtree
using any of the standard techniques~\cite{MR1-89,715896}.
\end{itemize}

Consider a graph $G$ with an edge weight function $w:E\to [0, \poly(n)]$.
For a vertex $v$ we denote by $G_v$ the graph $G-v$.
If $G$ has an even number of vertices then $M_v$ denotes some {\em minimum almost perfect matching} in $G_v$, i.e., a
minimum weight matching that misses exactly one vertex. If $G$ has an odd number of vertices then $M_v$ denotes some minimum
perfect matching in $G_v$. $M_v$ is not defined in unique way, but its weight $w(M_v)$ is. In our algorithms we
will only use these weights that can be computed using standard techniques as shown in Appendix~\ref{appendix-matching-weights}.

\begin{corollary}
\label{corollary-matching-weights}
For a graph $G=(V,E)$ with edge weights $w:E\to [0, \poly(n)]$ we can in NC:
\begin{itemize}
\item for a given vertex $v\in V$, compute the weight $w(M_v)$,
\item for a given edge $e\in E$, check whether $e$ is {\em allowed}, i.e., belongs to some minimum perfect matching.
\end{itemize}
\end{corollary}
Observe that the set of allowed edge is a subset of tight edges. Hence, when we remove all not allowed edges only tight edges are
left in the graph.

Let us now state the following implication of the results in~\cite{CGSa}. Basically, assuming that all $w(M_v)$ are given, Algorithm~5
from~\cite{CGSa} gives an NC procedure for computing the blossoms of the critical dual solution. For the
completeness of the presentation we have included this simple algorithm in Appendix~\ref{appendix-finding-blossoms}.

\begin{lemma}[based on Lemma 6.19~\cite{CGSa}]
\label{lemma-critical-dual}
Let $G=(V,E)$ be undirected connected graph where edge weights are given by $w:E\to \mathbb{Z}$
and where every edge is allowed.
Given all values $w(M_v)$ for $v \in V$, the blossoms of the critical dual solution can be computed in NC.
\end{lemma}

Algorithm~5 from~\cite{CGSa} actually constructs a critical dual
with an additional property which is called balanced. Intuitively, in a balanced dual the root of the
laminar tree is a central vertex of this tree as well. For formal definition see~\cite{CGSa} or see~\cite{Gabow:2012}
for an alternative definition of canonical dual. In short words, balanced duals are unique. Thus when one constructs them in parallel
as all processors construct the same solution, and the algorithm can be implemented in NC.

\section{The High Level Idea: Cycles and Blossoms}
This section aims to introduce two core ideas of our algorithm that allow us to reduce the size of the graph for the recursion, as well as a high level idea
that reveals around them. We will first give an algorithm for finding a perfect matching in a graph. However,
we will view the problem as weighted and seek minimum perfect matching.
This algorithm will be extended to weighted case in Section~\ref{section-minimum-perfect}.
The weighed view is useful as we will find even length cycles in the graph
and introduce weights on them. These weights will either cause some edges to become not allowed, or induce a blossom as shown by the following lemma.

To make it more precise we say that a cycle $C$ in graph $G$ is {\em semi-simple} if it
contains an edge that appears on $C$ only once. By $e_C$ we denote an arbitrary such edge on $C$.
\begin{lemma}
\label{lemma-even-cycle}
 Consider $G=(V,E)$ with edge weight function $w:E\to [0, \poly(n)]$. Let $C$ be an even semi-simple cycle in $G$. Let $w(e_C)=1$ and let $w(e)=0$ for all $e\in C-e_C$.
Then, either some edge of $C$ is not allowed or some edge of $C$ is in $\delta(B)$ for some blossom $B$.
\end{lemma}
\begin{proof}
Assume by a contradiction that all edges of $C$ are allowed and there is no blossom intersecting $C$. Hence, by complimentary slackness
conditions we have that $\pi_x+\pi_z=w(xz)$ for all edges $xz\in C$. In particular, for the edge $u_Cv_C=e_C$ we need to have $\pi_{u_C}+\pi_{v_C}=1$, whereas
for all other edges $uv\in C-e_C$ we have $\pi_u+\pi_v=0$. Now consider edge $zy\in C$ which is next to $xz$.
By subtracting equalities for edge $xz$ and $zy$ we obtain $\pi_x-\pi_y=w(xz)-w(zy)$. If $yz$ and $xz$ are not equal to $e_C$ we have $\pi_x-\pi_y=0$. In
general, we have $\pi_x-\pi_y=0$ for any two vertices at even distance along path $C-e_C$. And $\pi_x+\pi_y=0$
for any two vertices at odd distance along path $C-e_C$. Note that the distance from $u_C$ to $v_C$
along $C-e_C$ is odd, so we obtain $\pi_{u_C}+\pi_{v_C}=0$, what leads to contradiction.
See Figure~\ref{fig:double-faces}~a) for an illustration.
\end{proof}

Blossoms are natural objects to recurse on, as by duality slackness there must be
exactly one edge of any perfect matching $M$ that belongs to $\delta(B)$ for any blossom $B$. Thus, in the recursion,
we can find an almost perfect matching outside of $B$, an almost perfect matching inside of $B$ and then combine
them by matching one edge in $\delta(B)$ -- see Section~\ref{section-single-path}. However, having an edge in $\delta(B)$ does not directly guarantee that the size
of the graph reduces in the recursion. We need the following stronger observation for this.

\begin{lemma}
\label{lemma-two-sides}
Consider $G=(V,E)$ with edge weight function $w:E\to [0, \poly(n)]$. Let $C$ be an even semi-simple cycle in $G$. Let $w(e_C)=1$ and let $w(e)=0$ for all $e\in C-e_C$.
Then there exist edges $e_1,e_2\in C$, such that $e_1\in E(G\setminus B)$ and $e_2\in E(B)$.
\end{lemma}
\begin{proof}
For contradiction assume that no edge of $C$ is in $E(B)$ -- the other case is symmetric.
Hence, there can only be vertices $z$
in $E(B)$ such that its both incident edges $xz$ and $zy$ on $C$ are in $\delta(B)$.
In such a case we have $\pi_x+\pi_z+\pi_B=w(xz)$ and $\pi_z+\pi_y+\pi_B=w(zy)$. Thus
$\pi_x-\pi_y=w(xz)-w(zy)$, i.e., the contribution of $\pi_B$ when subtracting these equalities cancels.
See Figure~\ref{fig:double-faces}~b) for illustration. Thus $B$
does not contribute anything to the telescoping sum along $C-e_C$ and we reach similar contradiction as in the previous lemma.
\end{proof}

Hence, when recursing on the inside of the blossom or on the outside of the blossom we reduce the graph size as well. However, to obtain
an NC algorithm we need to reduce the size of the graph by a constant factor, so we need to have $\Omega(n)$ edge disjoint even cycles.
The following lemma, that was implicitly proven in~\cite{Kulkarni2004}, becomes handy now.
\begin{lemma}
\label{lemma-double-faces-x}
In a $2$-connected planar graph $G$ with $f$ faces we can find $\Omega(f)$ edge disjoint even semi-simple cycles in NC.
\end{lemma}
A simplified proof of this lemma is given in Appendix~\ref{appendix-finding-double-faces}. Intuitively, the proof
of this lemma puts faces into pairs that are incident. Now, either a pair contains an even face and thus this face is semi-simple, or both faces
are odd. In the later case we can built an even semi-simple cycle by walking around both faces.

The above lemma shows that to have many even semi-simple cycles we just need to guarantee that the graph has many faces. Let us
say that a planar graph $G$ is {\em simplified}, if there are no degree $1$ vertices and no two vertices of degree $2$ are incident.
The next lemma shows that such graphs have many faces.
\begin{lemma}
\label{lemma-degree-3}
Let $G=(V,E)$ be a simplified planar graph with a perfect matching, then $G$ has at least $\frac{n}{4}+2$ faces.
\end{lemma}
\begin{proof}
By Euler's formula $n-m+f=2$, where $m$ is the number of edges, and $f$ is the number of faces.
Let $V_2$ be the set of degree $2$ vertices in $G$. Consequently, the vertices in $V\setminus V_2$ have degree at least $2$.
As $G$ has a perfect matching and no two degree $2$ vertices are incident, all vertices
in $V_2$ need to be matched to vertices in $V\setminus V_2$, i.e., $|V_2|\le |V\setminus V_2|$. By using
this inequality we have.
\[
2m=\sum_{v\in V}\deg(v)=\sum_{v\in V_2}\deg(v) +  \sum_{v\in V\setminus V_2}\deg(v)\ge 2|V_2|+3|V\setminus V_2| \ge \frac{5}{2}|V_2| + \frac{5}{2}|V\setminus V_2|=\frac{5n}{2}.
\]
By plugging this inequality into Euler's formula we obtain $f=2+m-n\ge 2+\frac{5n}{4}-n=2+\frac{n}{4}$.
\end{proof}

Thus a simplified graph has many faces and many edge disjoint even cycles. We can assign weights to each of these cycles separately. This way
either many edges become not allowed, or there are many blossoms containing edges inside. The main technical part of the algorithm
is to handle family of blossoms. The algorithm has the following steps.
\begin{enumerate}
\item Simplify the graph as shown in Section~\ref{section:degree-2}. First, we
take care of degree $1$ vertices, by removing not allowed edges using Corollary~\ref{corollary-matching-weights} and matching independent edges. Next, we contract
paths composed of degree $2$ vertices. Finally, we find perfect matchings in each connected component separately. Due to removal of not allowed edges each
connected component is $2$-connected as well.
\item Using Lemma~\ref{lemma-double-faces-x} we find many edge disjoint even length cycles.
We assign weights to even cycles using Lemma~\ref{lemma-even-cycle} and we remove not allowed edges.
Next, we find blossoms of the critical dual using Lemma~\ref{lemma-critical-dual}. These steps are described in Section~\ref{section-main-routine}.
\item Finally, we recurse on a critical dual as explained in Section~\ref{section-single-path}, where we show how to construct
a perfect matching that respect all blossoms in the dual. This construction is recursive and calls back step 1 for subgraphs
of $G$ that do not have any blossoms.
\end{enumerate}
The recursion depth in this procedure is $O(\log n)$ as
either there $\Omega(n)$ edges become not allowed, or there are many blossoms. In the second case for each blossom $B_e$ there is an edge $e_1$ inside and an edge $e_2$ outside of it -- see Figure~\ref{fig:blossoms} a). In the case when we have many blossoms we observe that they divide the graph into regions and each lowest level recursive call goes to one
of these regions, i.e., we recourse on inside and outside of some blossom as long as there are some blossoms intersecting the current subgraph. As visualized on Figure~\ref{fig:blossoms} b)
there are many edges that are not incident to each region, as one of the edges from each pair $(e_1,e_2)$ needs be outside of this region.
Hence the size of each of these regions decreases by a constant factor with respect to the original graph. This will be proven more formally in Section~\ref{section-main-routine}
where we analyze the running time of the algorithm.

\section{Simplifying the Graph}
\label{section:degree-2}
The first ingredient of our algorithm is the following entry procedure that simplifies the graph and assures that we are working with $2$-connected graphs.
\begin{algorithm}[H]
\begin{algorithmic}[1]
\State{Remove all not-allowed edges from $G$.}
\State{Add independent edges to $M$ and remove them from $G$.}
\ForAll{paths $p$ composed out of degree $2$ vertices}\Comment{In parallel}
\If{$p$ has odd length}
\State{remove $p$ and connect vertices incident to $p$ with an edge $e_p$.}
\Else
\State{contract $p$ to a single vertex $v_p$.}
\EndIf
\EndFor
\State{Compute connected components $\cC$ of $G$}
\ForAll{components $C\in \cC$ of $G$}\Comment{In parallel}
\State{find perfect matching $M_C$ using Algorithm~\ref{algorithm-matching}.}
\EndFor
\State{Extend matching $M$ on paths of degree $2$ vertices.}
\State{return $M \cup \bigcup_{C\in \cC} M_C$.}
\end{algorithmic}
\caption{Simplifies graph $G$ and seeks perfect matchings on its $2$-connected components.}
\label{algorithm-degree-reduction}
\end{algorithm}

The following lemma proves the correctness of the above algorithm.
\begin{lemma}
Algorithm~\ref{algorithm-degree-reduction} executes Algorithm~\ref{algorithm-matching} on a simplified $2$-connected graph.
\end{lemma}
\begin{proof}
We first observe that after removing not allowed edges, an edge $e$ incident to a vertex of with degree $1$ needs to belong to a
perfect matching. Moreover, $e$ is independent. Hence, after matching such independent edges all vertices have degree $2$ or higher.
Now note that the manipulation of the path $p$ does not change degrees of other vertices.
In the algorithm such path is either replaced by an edge, or single vertex -- see Figure~\ref{fig:degree-2}. Hence, afterwards
degree $2$ vertices are independent. Thus $G$ is simplified and its all connected components are simplified as well.

Now, for contrary assume that a connected component $C$ is not $2$-connected, i.e., there exists an articulation point $v$. Consider connected components
obtained from $C$ after removal of $v$. Only one of them can be matched in the perfect matching to $v$. Thus only this one
has odd number of vertices. The remaining components must have even number of vertices and no perfect matching can match them to $v$.
Hence, their edges incident to $v$ are not allowed. However, all not allowed edges were removed by the algorithm, so we
reach a contradiction.
\end{proof}
In the final step of the algorithm we need to expand all paths that were replaced during the execution of Algorithm~\ref{algorithm-degree-reduction}.
Observe that the matching $M$ in the simplified graph can be extended to the matching in the orginal graph in a straight-forward way. If $p$ was odd, then
depending on whether $e_p$ is matched we either match even or odd edges on $p$. If $p$ was even, then $v_p$ has degree $2$ and it can be matched in
one of two possible ways. In these two cases the matching can be extended to the whole path.

\section{The Main Routine}
\label{section-main-routine}

Algorithm~\ref{algorithm-matching} implements the main procedure of the algorithm. First, we
 find $\Omega(n)$ even semi-simple cycles and introduce weights on them. Next, in order to reduce the size of the graph we remove not allowed edges.
We then find blossoms of the critical dual solution with respect to these weights and call Algorithm~\ref{algorithm-matching-recurse} to find
a perfect matching that respects all blossoms.

\begin{algorithm}[H]
\begin{algorithmic}[1]
\State{Find a set $\mathcal{F}$ of $\Omega(n)$ edge disjoint even semi-simple cycles using Lemma~\ref{lemma-double-faces-x}.}
\State{Set $w(e)=0$ for all $e\in E$.}
\ForAll{$C \in \mathcal{F}$}\Comment{In parallel}
\State{Set $w(e_C)=1$.}
\EndFor
\label{algorithm-matching-f}
\State{Remove all not-allowed edges from $G$.}
\State{Compute blossoms of a critical dual $\cB$ with respect to $w$ using Lemma~\ref{lemma-critical-dual}.}
\State{Compute a matching $M$ that respects $\cB$ using Algorithm~\ref{algorithm-matching-recurse}.}
\State{Return $M$.}
\end{algorithmic}
\caption{Finds a perfect matching $M$ in a connected graph $G$. If the graph is not connected we call the procedure
for each component separately.}
\label{algorithm-matching}
\end{algorithm}


\section{Finding a Perfect Matching that Respects a Family of Blossoms}
\label{section-single-path}
Let $G=(V,E)$ be a graph and let $M$ be a matching in $G$. An {\em alternating path} $p$ is a path in $G$
such that edges on $p$ alternate between matched and unmatched. Assume that $G$ is factor critical (inside of a blossom)
and that $M_s$ is an almost perfect matching that misses vertex $s$. We start by showing
how using $M_s$ we construct an almost perfect matching $M_v$ for any $v$. To this end, we need
to find a simple alternating path starring in $s$ and ending in $v$. Denote by $w_M$  a weight function assigning $0$ to edges in $M_s$ and $1$ to edges not in $M_s$.

\begin{lemma}
\label{lemma-short-paths}
Let $M_v$ be the minimum almost perfect matching in $G_v$ with respect to $w_M$, then
$2w_M(M_v)$ is the length of the shortest alternating path with respect to $M_s$ from $s$ to $v$.
\end{lemma}
\begin{proof}
Observe that the symmetric difference $M_s \oplus M_v$ contains an alternating path with respect to $M_s$ that
needs to start at $s$ and end at $v$. The weight of this path is equal to the number of edges from $M_s$ on it. As this
path is alternating the number of edges of $M_v$ is the same. Thus minimizing $w_M(M_v)$ we minimize the length of
an alternating path with respect to $M_s$ from $s$ to $v$.
\end{proof}

Now, we want to construct a graph $G_L$ that will represent
all shortest alternating paths from $s$ with respect to $M$. $G_L$ will be a layered graph, where
layer $l$ contains vertices at distance $l$ from $s$ --- the distance is measured along alternating paths.
For each $v\in V$, let $G_L$ contain two copies $v^o$ and $v^e$ of $v$.
We define $l(v^e)= 2w_M(M_v)$ for all $v\in V$, and $l(v^o)= l(u^e)-1$ for all $uv\in M$.
We add edges of $G$ to $G_L$ only if they connect two consecutive layers given by $l$ -- see Figure~\ref{fig:single-path}.
Every shortest alternating path from $s$ is contained in $G_L$ by Lemma~\ref{lemma-short-paths}. Alteratively, if a path in
$G_L$ represents a simple path in $G$ then it is a shortest alternating path in $G$. However, there are
path in $G_L$ that do not correspond to simple paths in $G$, i.e., they contain both $v^o$ and $v^e$ for some $v\in V$ -- see Figure~\ref{fig:single-path} b).
Nevertheless, as every vertex in $G$ is reachable via alternating path, we can modify $G_L$ in such a way that only path corresponding to simple path in $G$
remain. This is done using Algorithm~\ref{algorithm-fix-up}.
\begin{algorithm}[H]
\begin{algorithmic}[1]
\State{For all $v\in V$ compute $w_M(M_v)$.}
\State{Let $G_L$ be a graph where $v$ has two copies $v^o$ and $v^e$.}
\State{For all $v\in V$ set $l(v^e)= 2w_M(M_v)$.}
\State{For all $vu \in M$ set $l(v^o)= l(u^e)-1$.}
\State{Add edges of $G$ to $G_L$ only if they connect vertices in consecutive layers $l$.}
\ForAll{$v\in V$}\Comment{In parallel}
\ForAll{$u^z$ on some $v^x$-$v^y$ path in $G_L$, where $x,y,z\in \{e,o\}$, $x\neq y$}\Comment{In parallel}
\If{there exists $s^e$-$u^z$ path $p$ avoiding $v^x$ in $G_L$}
\State{remove all edges entering $u^z$ but the edge on $p$.}
\EndIf
\EndFor
\EndFor
\State{Return any path from $s^e$ to $t^e$ in $G_L$.}
\end{algorithmic}
\caption{Finds an alternating path with respect to $M_s$ in $G$ from a vertex $s$ to vertex $t$.}
\label{algorithm-fix-up}
\end{algorithm}

The correctness of this algorithm is established by the next lemma.

\begin{lemma}
\label{lemma-fix-up}
Let $G$ be a factor critical graph and let $M_s$ be an almost perfect matching missing vertex $s$.
An almost perfect matching $M_u$ missing vertex $u$ can be found in NC using Algorithm~\ref{algorithm-fix-up}.
\end{lemma}
\begin{proof}
We first observe that the removal of edges entering $u^z$ from $G_L$ does not affect reachability from $s$, as
the path $p$ from $s^e$ to $u^z$ is left in the graph. Now, by contradiction, assume that at the end
of the algorithm there is a path $q$ in $G_L$ that contains both $v^e$ and $v^o$ for some $v\in V$.
Without loss of generality assume $v^e$ precedes $v^o$. As $G_L$ contains all simple alternating paths, there exists an $s^e$-$v^o$ path
$p$ avoiding $v^e$. Hence, the edge of $q$ entering the first shared vertex with $p$ was removed by the algorithm -- see Figure~\ref{fig:single-path}~b).\footnote{The
graph constructed in this algorithm can be seen as an extended version of generalized shortest path tree~\cite{6686149}. Alternatively, such
tree could be constructed using Algorithm 5 from~\cite{CGSa}. This, however, would result in a slightly more complicated solution.}
\end{proof}

The next algorithm constructs a perfect matching $M$ that respects a family of blossoms.
Here, we explicitly consider $\cB$ as a tree $T_{\cB}$, i.e., the vertices of $T_{\cB}$ are sets in $\cB$
whereas edges in $T_{\cB}$ represent child parent relationship. See Figure~\ref{fig:recurse-tree} for an example
and an illustration of recursion. It calls Algorithm~\ref{algorithm-degree-reduction} that handles the case without blossoms.
The next theorem argues about the correctness of this algorithm.
\begin{lemma}
\label{lemma-matching-recurse}
Algorithm~\ref{algorithm-matching-recurse} finds a perfect matching $M$ respecting $\cB$.
The recursion depth of the internal calls of the algorithm to itself is $O(\log n)$.
\end{lemma}
\begin{proof}
We need to argue that $M_B$ can be extended to a perfect matching in the whole graph. Consider a child blossom $C$ of $B$. By Lemma~\ref{lemma-critical-dual}
we know that $C$ is factor-critical, so there exists almost perfect matching $M'_{C}$ in $C$ that together with $M_B$ forms a perfect matching.
This matching differs from $M_C$ by a single alternating path. Moreover, note that after contraction of a given blossom all nonintersecting blossoms remain
blossoms and the graph remains planar, so we can continue recursing on subtrees of $T_{\cB}$. As we recurse on a vertex separator of a tree $T_{\cB}$ the size
of the subtrees decreases by a constant factor.
\end{proof}
Lemma~\ref{lemma-matching-recurse} leads to the correctness of Algorithm~\ref{algorithm-matching} and Algorithm~\ref{algorithm-degree-reduction} as well.
\begin{theorem}
Algorithm~\ref{algorithm-degree-reduction} finds a perfect matching in $G$.
\end{theorem}

\begin{algorithm}[H]
\begin{algorithmic}[1]
\If{$\cB=\emptyset$}
\State{Return matching $M$ computed by Algorithm~\ref{algorithm-degree-reduction} on $G$.}
\EndIf
\State{Let $B$ be a vertex separator of $T_{\cB}$.}
\Begin\Comment{In parallel with the next loop}
\State{Let $G_B$ be the graph with all children of $B$ contracted.}
\State{Let $T_B$ be $T_{\cB}$ with children of $B$ removed.}
\State{Recurse on $G_B,T_B$ to obtain matching $M_B$.}
\End
\ForAll{children $C$ of $B$ in $T_{\cB}$}\Comment{In parallel}
\State{Let $G_C$ be $G$ with all vertices not in $C$ contracted to a vertex denoted by $v_C$.}
\State{Let $T_C$ be subtree of $T_{\cB}$ rooted at $C$.}
\State{Recurse on $G_C,T_C$ to obtain matching $M_C$.}
\EndFor
\ForAll{children $C$ of $B$ in $T_{\cB}$}\Comment{In parallel}
\State{Remove from $G_C$ vertex $v_C$ and let $u_C$ be the resulting free vertex in $G_C$.}
\State{Let $e_B$ be the edge of $M_B$ incident to $G_C$.}
\State{Let $v_B$ be the endpoint of $e_B$ in $G_C$ after expanding $G_C$.}
\State{Apply to $M_C$ an alternating path from $u_C$ to $v_B$ found using Algorithm~\ref{algorithm-fix-up}}
\EndFor
\State{Return $M_B\cup \bigcup_{C \textrm{ child of }B} M_C$.}
\end{algorithmic}
\caption{Computes a perfect matching $M$ of $G$ respecting blossoms $\cB$.}
\label{algorithm-matching-recurse}
\end{algorithm}

We now can turn our attention to arguing about the running time of our algorithm. In this section we are going to quantify progress
related to perturbing weights on each semi-simple cycle.

%


\begin{lemma}
\label{lemma-recursion-depth}
The number of edges in $G$ decreases by a constant factor when recursing to Algorithm~\ref{algorithm-degree-reduction} from Algorithm~\ref{algorithm-matching} via Algorithm~\ref{algorithm-matching-recurse}.
\end{lemma}
\begin{proof}
By Lemma~\ref{lemma-double-faces-x} and Lemma~\ref{lemma-degree-3} after graph simplification we have $\Omega(n)$ even semi-simple cycles in the graph $G$.
Now, by Lemma~\ref{lemma-even-cycle} for each semi-simple cycle, either one edge becomes not-allowed, or there exists a blossom $B\in \cB$ that intersects
this cycle. Hence, either $\Omega(n)$ edges become not-allowed, or we have $\Omega(n)$ cycles intersected
by some blossom. Lemma~\ref{lemma-two-sides} implies that for each such intersection there exists an edge $e_1$ inside $B$ and an edge $e_2$ outside
$B$ -- see Figure~\ref{fig:blossoms}~a). Consider a plane embedding of $G$ and draw boundaries of each blossom $\delta(B)$ in this plane, thus dividing the plane and graph $G$ into regions $\mathcal{R}$.
The outside of each region $R\in \mathcal{R}$ contains at least one edge from $e_1$ and $e_2$, i.e., there are $\Omega(n)$ edges not incident to a region -- see Figure~\ref{fig:blossoms}~b).
We note that when we recurse to Algorithm~\ref{algorithm-degree-reduction} from Algorithm~\ref{algorithm-matching-recurse}, we recurse
onto some region $R\in \mathcal{R}$ with parts of $G$ not in $R$ contracted to vertices. This graph contains only edges incident to a region $R$, so
it does not contain $\Omega(n)$ edges. By Euler's formula the total numer of edges is $\le 3n-6$, so when recursing on each region it decreases by a constant factor.
\end{proof}

By Lemma~\ref{lemma-matching-recurse} and~\ref{lemma-recursion-depth} the recursion depth in Algorithm~\ref{algorithm-matching} is $O(\log^2 n)$ thus:
\begin{corollary}
A perfect matching in a planar graph can be computed in NC.
\end{corollary}

\section{Minimum Perfect Matching}
\label{section-minimum-perfect}
So far we have assumed that the graph is unweighed and we have coped with the problem of constructing any perfect matching.
However, our approach is versatile enough to handle weighted case in a rather straightforward way using Algorithm~\ref{algorithm-weighted-matching}.

\begin{algorithm}
\begin{algorithmic}
\State{Remove all not allowed edges from $G$.}
\State{Compute blossoms of a critical dual $\cB$ with respect to $w$ using Lemma~\ref{lemma-critical-dual}.}
\State{Find a perfect matching $M$ respecting $\cB$ using Algorithm~\ref{algorithm-matching-recurse}.}
\State{Return $M$.}
\end{algorithmic}
\caption{Finds a minimum perfect matching in $G=(V,E)$ with respect to the edge weight function $w:E\to [0, \poly(n)]$.}
\label{algorithm-weighted-matching}
\end{algorithm}
Hence, we obtain the following.
\begin{lemma}
\label{lemma-minimum-perfect-matching}
A minimum perfect matching in a planar graph $G=(V,E)$ with edge weight function $w:E\to [0, \poly(n)]$ can be computed in NC.
\end{lemma}
\begin{proof}
Observe that all allowed edges in $G$ need to be tight. By complimentary slackness conditions a perfect matching that is composed out of allowed edges and that respects all
blossoms is a minimum perfect matching.
\end{proof}


\section{Minimum $f$-Factors}
\label{section-f-factor}
Let $G=(V,E)$ be a multi-graph, i.e., we allow parallel edges as well as self-loops.
 For a function $f:V\to [1,\poly(n)]$, an {\em $f$-factor} is a set of edges $F\subseteq E$
such that $\deg_F(v)=f(v)$ for every $v\in V$. Without loss of generality  we assume that
any edge $uv$ has multiplicity at most $\min \{f(u),f(v)\}$.
A {\em minimum $f$-factor} is an $f$-factor $F$ with minimum weight $w(F)$.

The usual approach to $f$-factor problems is by vertex-splitting~\cite{Schrijver,hal83,anstee}, but
these reductions do not preserve planarity. Here, we provide planarity
preserving vertex splitting. In particular, we show how to
replace each vertex $v\in V$ with a planar gadget $\tilde{G}$, such
that perfect matchings in the resulting graph correspond to $f$-factors
in the original graph.

The gadget $\tilde{G}^{d,f}$ is parameterized by the degree $d$ of a vertex
and its $f$ value. In order to work, the gadget $\tilde{G}^{d,f}$ needs
to have $d$ interface vertices $I=\{i_1,,,i_d\}$ ordered in a circular order on
its outside face. We require that for every subset of $I'\subseteq I$ the
graph $G\setminus I'$ has a perfect matching when $|I'|=f$, whereas
$G\setminus I'$ has no perfect matching when $|I'|\neq f$.

The construction of $\tilde{G}^{d,f}$ is done recursively -- see Figure~\ref{fig:f-factors}~a). We start the recursion for $d=f$, and let $\tilde{G}^{d,f}$
be a set of $d$ independent vertices, where all of them form the interface. The graph $\tilde{G}^{d+1,f}$ is defined
from $\tilde{G}^{d,f}$ with interface $I=\{i_1,\ldots,i_d\}$ in the following way:
\begin{itemize}
\item add $d+1$ vertices $v_1,\ldots,v_{d+1}$ and connect $i_j$ to both $v_j$ and $v_{j+1}$ for all $j\in [1, d]$,
\item add $d+1$ vertices $I'=\{i'_1,\ldots,i'_{d+1}\}$ and connect $i'_j$ with $v_j$ for all $j\in  [1, d]$.
\item $I'$ is the interface of $\tilde{G}^{d+1,f}$.
\end{itemize}

We note that $\tilde{G}^{d,f}$ has $O((d-f)d)=O(d^2)$ vertices
and edges. Hence, by exchanging each vertex with such gadget we obtain a graph of polynomial size. Moreover,
if the graph is weighted, we set weights of edges in all gadgets to $0$. This way the weight of the
resulting perfect matching is equal to the weight of the $f$-factor.

\begin{lemma}
\label{lemma-minimum-f-factor}
Let $G=(V,E)$ be a planar multigraph with edge weight function $w:E\to [0, \poly(n)]$. For a function $f:V\to [1,\poly(n)]$, minimum $f$-factor can be computed in NC.
\end{lemma}
\begin{proof}
Take graph $G$ and replace each vertex $v\in V$ with gadget $\tilde{G}^{\deg(v),f(v)}$ connecting incident edges to its interfaces.
The resulting graph has $\sum_{v\in V} O(\deg(v)^2) = O(n^3)$ vertices and the minimum perfect matching in it corresponds to minimum $f$-factor
in $G$. Applying Lemma~\ref{lemma-minimum-perfect-matching} to this graph finishes the proof.
\end{proof}
We note that maximum $f$-factor (or prefect matching) can be computed by using redefined weight function $w'(e)=-w(e)+\max_{f\in E}w(f)$.


\section{Maximum Bipartite Matching}
\label{section-maximum-bipartite}
Let $G=(V,E)$ be a bipartite planar graph with edge weight function $w:E\to [0, \poly(n)]$.
Algorithm~\ref{algorithm-maximum-bipartite-matching} computes a maximum matching in graph $G$, i.e., a matching of maximum total weight.
A $2$-factor is an $f$-factor when $f(v)=2$ for all $v\in V$.
\begin{algorithm}
\begin{algorithmic}
\State{Create a multigraph $G'$ from $G$ by taking two copies of each edge $E$.}
\State{Add a self loop $vv$ of weight $0$ for each $v\in V$.}
\State{Compute a maximum $2$-factor $F$ in $G'$.}
\State{Remove all self loops from $F$.}
\ForAll{even length cycles $C$ in $F$}\Comment{In parallel}
\State{Remove every second edge from $C$.}
\EndFor
\ForAll{edges $e \in E$ taken twice in $F$}\Comment{In parallel}
\State{Remove one copy of $e$ from $F$.}
\EndFor
\State{Return $F$.}
\end{algorithmic}
\caption{Finds a maximum matching in a bipartite graph $G=(V,E)$ with edge weight function $w:E\to [0,\poly(n)]$.}
\label{algorithm-maximum-bipartite-matching}
\end{algorithm}

Keep in mind that $G$ is bipartite, so $F$ cannot contain odd length cycles (besides the self-loops), so processing them is not needed in the
above algorithm. Note that $G'$ is not bipartite as it contains self loops.

\begin{lemma}
A maximum matching in a planar bipartite graph $G=(V,E)$ with edge weight function $w:E\to [0, \poly(n)]$ can be computed in NC.
\end{lemma}
\begin{proof}
As $f(v)=2$ for all $v\in F$, we can apply Lemma~\ref{lemma-minimum-f-factor} to obtain $2$-factor $F$ in NC.
Because $G$ is bipartite, after removal of self-loops $F$ contains only even cycles or edges taken twice. Thus the algorithm
constructs a matching $M$ out of it. The weight of $M$ is equal to half the weight
of $F$, as both ways of choosing edges from even length cycle need to have the same weight.
To prove optimality of $M$ assume there exists a matching $M^*$ with bigger weight.
Take $M^*$ twice and add self loop to all free vertices. This would give an $2$-factor with bigger
weight than $F$.
\end{proof}

By setting $w(e)=1$ for all $e\in E$, we can compute maximum size matching in a bipartite graph. We can make the above
reduction of maximum bipartite matching to maximum (non-bipartite) perfect matchings slightly stronger by guaranteeing that the size of the resulting graph
remains $O(n)$. The only requirement is to first reduce the maximum degree of a graph to $3$ using standard reduction (see e.g.,~\cite{esa}).
In this way $f$-factor gadgets will have constant size and the resulting graph has $O(n)$ size.

\bibliography{weighted}

\appendix

\section{Proof of Corollary~\ref{corollary-matching-weights}}
\label{appendix-matching-weights}
Consider a planar graph $G=(V,E)$ and a sign function $s:E\to \{-1,1\}$ for edges.
Let us define a \emph{signed adjacency matrix} of graph $G$ to be the $n \times n$ matrix $A(G,s)$
such that:
\[
A(G,s)_{i,j} =
 \left\{
 \begin{array}{rl}
 s(ij) & \textrm{if  } {ij} \in E \textrm{ and } i < j, \\
 -s(ij) & \textrm{if  } {ij} \in E \textrm{ and } i > j, \\
 0 & \textrm{otherwise.}
\end{array}
\right.
\]

\begin{theorem}[Kasteleyn~\cite{kasteleyn}]
There exists a function $s:E\to \{-1,1\}$ such that $\sqrt{\det(A(G,s))}$ is equal to the number of perfect matchings in a planar graph $G$.
Such function $s$ is called {\em Pfaffian orientation}.
\end{theorem}
Let us explain shortly the idea behind this theorem. Consider the definition $\det(A(G,s))=\sum_{\pi \in \Pi_n} \sigma(\pi) \prod_{i=1}^{n}A(G,s)_{i,\pi(i)}$, where
$\Pi_n$ is the set of all $n$ element permutations and $\sigma$ denotes a sign of a permutation. Each term $\sigma(\pi) \prod_{i=1}^{n}A(G,s)_{i,\pi(i)}$
in $\det(A(G,s))$ can be seen as a set of edges $M_2=\{i\pi(i)\}_1^n$ in $G$. This term will contribute non-zero to the determinant only if all cycles in $M_2$ have
even length. If $M_2$ contains an odd cycle we can reverse signs $s$ on this cycle obtaining a term that will cancel out.
As shown in~\cite{Jerrum2003}, there is a bijection between ordered pairs of perfect matchings and all possible sets $M_2$ that contain cycles of even length,
i.e., non-zero terms in the determinant. The above theorem holds because for Pfaffian orientation all these terms have the same sign.

In our algorithms we will never directly need the number of matchings, but we need this idea to compute the weight of a minimum perfect matching.
Let $w:E\to [0, \poly(n)]$ be integral edge weight function.
We define a \emph{signed weighted adjacency matrix} of graph $G$ to be the $n \times n$ polynomial matrix $A(G,s,w)$
such that:
\[
A(G,s,w)_{i,j} =
 \left\{
 \begin{array}{rl}
 s(ij)y^{w(ij)} & \textrm{if  } {ij} \in E \textrm{ and } i < j, \\
 -s(ij)y^{w(ij)} & \textrm{if  } {ij} \in E \textrm{ and } i > j, \\
 0 & \textrm{otherwise,}
\end{array}
\right.
\]
where $y$ is the variable of the polynomial. The terms of $\det(A(G,s,w))$ correspond to pairs of perfect matching and in each term the power of $y$ will
correspond to the sum of their weights. Hence, the terms in $\det(A(G,s,w))$ of minimum degree needs to correspond to pairs of minimum weight perfect matchings.
Hence, we obtain the following.
\begin{corollary}
For a Pfaffian orientation $s$, the degree of the minimum degree term of $\det(A(G,s,w))$ in $y$ is
equal to twice the weight of minimum weight perfect matching in $G$.
\end{corollary}

The determinant of matrix with univariate polynomials of degree $\poly(n)$ can be computed in NC as shown in~\cite{BORODIN1983}.\footnote{The exact statement of the theorem in~\cite{BORODIN1983} limits
the degree to $n$, but we can always pad the matrix with $1$'s on the diagonal to make its size equal to the degree.} Hence, we obtain the following.
\begin{corollary}
\label{corollary-number}
Given a planar graph $G=(E,V)$ and a weight function $w:E\to [0, \poly(n)]$ the weight $w(M^*)$ of minimum perfect matchings in $G$ can be computed in NC.
\end{corollary}

Now let us argue how to realize Corollary~\ref{corollary-matching-weights}. In order, to check whether an edge $e\in E$ is allowed
we just need to check whether $w(M_e)=w(M)$, where $M$ is the minimum perfect matching in $G$, and $M_e$ is the minimum perfect matching using edge $e$.
We have that $w(M_e)=w(M'_e)+w(e)$, where $M'_e$ is the minimum perfect matching in $G-e$, i.e., in $G$ with both endpoints of $e$ removed.
Finally, let us describe how to compute $w(M_v)$. If $|V|$ is odd, then $M_v$ is simply a minimum perfect matching of $G-v$. Otherwise,
when $|V|$ is even, we take the minimum of $w(M_{vu})$ over all $u\in U$, where $M_{uv}$ is a minimum perfect matching of $G-v-u$.


\section{Finding $\Omega(n)$ Edge Disjoint Even Semi-Simple Cycles}
\label{appendix-finding-double-faces}
In this section we assume to be working with a $2$-connected planar graph $G$. It is a standard assumption that implies that all faces of $G$ are simple.
The main artifacts, that will be useful here, are double faces defined in the following way.\footnote{Double faces correspond to Building Block 3 or 4 in~\cite{Kulkarni2004}.}

\begin{definition}
{\em Double face} $(f_1,f_2,p)$ is a pair of two simple faces $f_1$ and $f_2$, $f_1\neq f_2$, connected by a simple path $p$ possibly of length $0$.
The faces can be incident and even share edges. In such a case, $p$ contains a single vertex that is shared by both faces.
\end{definition}

If any of the faces of the double face is of even length then it immediately gives an even semi-simple cycle.
The next lemma argues about the case when both faces are odd.

\begin{lemma}
\label{lemma-double-face}
Let $(f_1,f_2,p)$ be an odd double face of $G$, i.e., both faces are of odd length, then there exists an even semi-simple cycle $C$ in $(f_1,f_2,p)$.
\end{lemma}
\begin{proof}
First, let us consider the case when $f_1$ and $f_2$ share an edge, and let $e_C$ be an edge of $f_1$ that does not belong to $f_2$.
Staring from $e_C$ walk along $f_1$ in both directions till you encounter vertices $u$ and $v$ that belong to both $f_1$ and $f_2$. Let $f_1[u,v]$ be
that part of $f_1$ containing $e_C$. Let $f_2[u,v]$ and $f_2[v,u]$ be the parts of $f_2$ that end at $u$ and $v$. Because $f_2$ has
odd length $f_2[u,v]$ and $f_2[v,u]$ have different parity, e.g., $f_2[u,v]$ has odd length, whereas $f_2[v,u]$ has even length.
If $f_1[u,v]$ has odd length we obtain odd cycle by joining it with $f_2[u,v]$. Otherwise, joining $f_1[u,v]$ with $f_2[v,u]$ gives an
odd cycle. In any case this cycle contains edge $e_C$ -- see Figure~\ref{fig:double-faces}~c).

Second, let us consider the case when $f_1$ and $f_2$ do not share an edge. In this case we can construct an odd length semi-simple cycle
by walking along $f_1$, $f_2$ and along both directions of $p$. This cycle is semi-simple, and we can take any edge of $f_1$ or $f_2$ as $e_C$.
\end{proof}

Let us now show how to find double faces in a graph. The \emph{dual} $G^*$ of $G$ is a multigraph having a vertex for each  face of $G$. For each edge $e$ in $G$,
there is an edge $e^*$ in $G^*$ between the vertices corresponding to the two faces of $G$ adjacent to $e$.
We identify faces of $G$ with vertices of $G^*$ and since there is a one-to-one correspondence
between edges of $G$ and edges of $G^*$, we identify an edge of $G$ with the corresponding edge in $G^*$.

The next algorithm matches faces of the graph into double faces -- see Figure~\ref{fig:double-faces-2}. The idea is to use a spanning tree $T_D$ of the dual graph of $G$. In such
a spanning tree we pair together two children of a node, or we pair a child with its parent. In the first case, we have a path connecting
two children, whereas in the second case two faces share an edge. Using such pair, we can match every face, but the root if the number of faces is odd.
Then we take maximal set of double faces that are edge disjoint. In Lemma~\ref{lemma-number-of-double-cycles} we show that there are $\Omega(f)=\Omega(n)$ such faces.

\begin{algorithm}[H]
\begin{algorithmic}[1]
\State{Find a spanning tree $T_D$ of the dual graph and root it at arbitrary vertex $f_r$.}
\State{For all $f\in T_D$ compute the number of vertices $n(f)$ in the subtree of $T_D$ rooted at $f$.}
\ForAll{$f\in T_D$}\Comment{In parallel}
\State{Let $c_f:=\{c\in T_D: c \textrm{ is a child of } f \textrm{ and } n(c) \textrm{ is odd}\}$.}
\If{$|c_f|$ is odd}
\State{Let $c$ be any child in $c_f$.}
\State{Mark $(c,f,\{v\})$ as a double face where $v$ is any vertex shared by $c$ and $f$.}
\State{$c_f:=c_f-c$.}
\EndIf
\State{Order $c_f$ according to the order around $f$.}
\State{Mark pairs of consecutive children $c_1,c_2$ together with path $p_f$ of $f$ that connects them as a double face $(c_1,c_2,p_f)$.}
\EndFor
\State{Construct a graph $G_{DF}$ where vertices are marked double faces and edges denote which double faces share an edge.}
\State{Find maximal independent set of vertices in $G_{DF}$ and return it.}
\end{algorithmic}
\caption{Given graph $G$ where all faces are simple, computes a set of $\Omega(n)$ edge disjoint double faces.}
\label{algorithm-double-cycles}
\end{algorithm}

\begin{figure}[t]
\centering
    \includegraphics[width=0.99\textwidth]{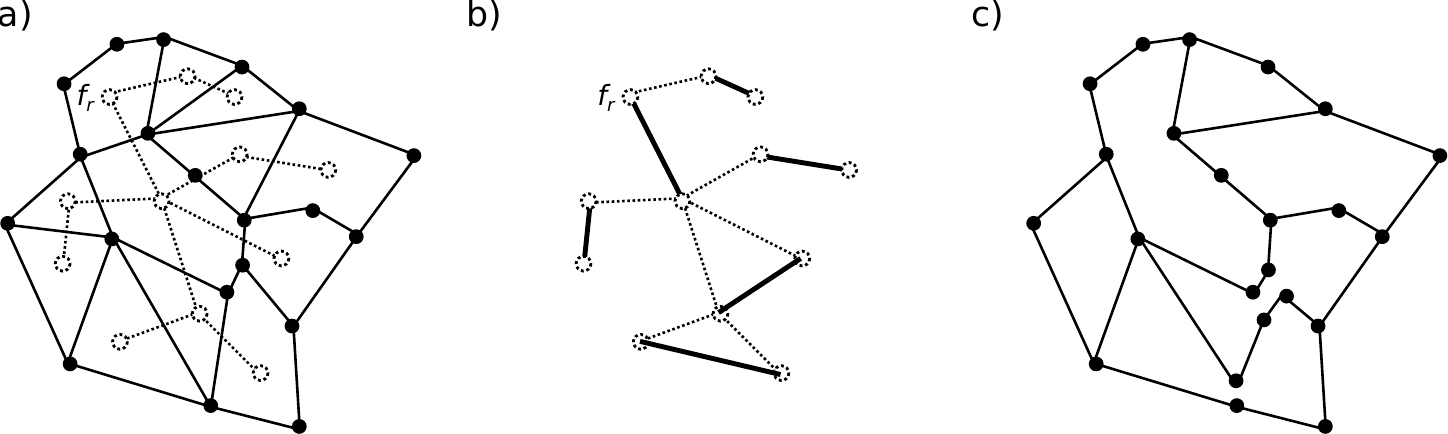}
    \caption{Figure a) shows a planar graph together with a spanning tree of the dual graph. Figure b) shows
    a paring of faces constructed by Algorithm~\ref{algorithm-double-cycles}. Figure c) presents a cut-open graph as
    constructed in the proof of Lemma~\ref{lemma-number-of-double-cycles} -- each double face corresponds to a face in this graph.\label{fig:double-faces-2}}
\end{figure}

\begin{lemma}[Lemma 3 from~\cite{Mahajan:2000}]
\label{lemma-mahajan}
A planar graph with $f$ faces contains a set of $\Omega(f)$ edge disjoint faces.
\end{lemma}

\begin{lemma}
\label{lemma-number-of-double-cycles}
The maximal independent set of vertices in $G_{DF}$ is of size $\Omega(n)$, i.e., Algorithm~\ref{algorithm-double-cycles} finds $\Omega(n)$ edge disjoint double faces.
\end{lemma}
\begin{proof}
Consider the graph $G$ and all the marked double faces. Let us construct a new graph $G'$ by converting double faces $(c_1,c_2,p_f)$ into faces in the following way.
If $c_1$ and $c_2$ share some edges we remove these edges. Otherwise, we cut the graph open along $p_f$ joining $c_1$ and $c_2$ into single face -- see Figure~\ref{fig:double-faces-2} c).
By the construction this graph has at least $\frac{f-1}{2}=\Omega(f)$ faces that correspond to double faces in $G$. The face $f_r$ might have remained not assigned to any double face. By Lemma~\ref{lemma-mahajan} we know that $G'$ contains $\Omega(f)$ edge disjoint faces. Finally, by Lemma~\ref{lemma-degree-3} we know that $f=\Omega(n)$. This ends the proof.
\end{proof}

This implies Lemma~\ref{lemma-double-faces-x}.

\section{Finding Blossoms}
\label{appendix-finding-blossoms}
In this section we show how to find a laminar family of blossoms that forms a critical dual. Algorithm~\ref{algorithm-blossoms} is essentially a restatement
of Algorithm~5 from~\cite{CGSa}. The only difference is that in~\cite{CGSa} the algorithm is stated as processing values
$\alpha$ in increasing order to save on sequential work. However, here we process each value independently in parallel.

We have to note that this algorithm actually constructs a critical dual
with an additional property which is called balanced.\footnote{For formal definition see~\cite{CGSa} or see~\cite{Gabow:2012}
for alternative definition of canonical dual.} The take out note is that this property makes the dual unique and explains
why this algorithm can be implemented in NC in a straightforward way.

\begin{algorithm}[H]
\begin{algorithmic}[1]
\State For each edge $uv$ set $w'(uv)=w(uv)+w(M_u)+w(M_v)$.
\State{Let $A$ be the set of all different values $w'(uv)$. Let $\mathcal{B}=\emptyset$.}
\For{each $\alpha \in A$}\Comment{In parallel}
\State{Let $\mathcal{C}$ be the set of connected components
  of the graph $(V,\{uv : uv \in E, w'(uv) \le \alpha\})$.}
\State Add the nontrivial components of $\mathcal{C}$ to $\mathcal{B}$.
\EndFor
\State \Return $\mathcal{B}$.
\end{algorithmic}
\caption{Given all the values $w(M_u)$, finds the blossoms of a balanced critical dual in the graph $G$ where all edges are allowed.}
\label{algorithm-blossoms}
\end{algorithm}

\begin{lemma}[Lemma 6.19~\cite{CGSa}]
Let $G=(V,E)$ be undirected connected graph where edge weights are given by $w:E\to [1,\poly(n)]$
and where every edge is allowed.
Given all values $w(M(v))$ for $v \in V$,
Algorithm~\ref{algorithm-blossoms} finds a blossoms of a balanced critical dual solution.
\end{lemma}

\section{Maximum Non-Bipartite Matching}
\label{appendix-maximum-non-biparitte}
In this section we will turn our attention to the most challenging problem -- maximum size non-bipartite matching, where
we will show an $O(\sqrt{n}\log n)$ time algorithm. Thus solving this problem in NC remains an intriguing open problem. We note
that so far no $o(n)$ time deterministic parallel algorithm for this problem was known.
We will first observe that Algorithm~\ref{algorithm-fix-up} can be adopted to the case when graph $G$ is not factor critical,
i.e., consider arbitrary almost perfect matching $M$ in $G$. For all $v\in V$, we redefine $l$ as:
\[
l(v^e)= \begin{cases} 2w_M(M_v)& \textrm{ if } M_v \textrm{ exists}, \\
\infty & \textrm{ otherwise}.
\end{cases}.
\]
Similarly, $l(v^o)= l(u^e)-1$ for all $uv\in M$, where $\infty-1=\infty$. The following is immediate
consequence of the proof given in Section~\ref{section-single-path}.

\begin{corollary}
\label{corollary-fix-up}
Let $G$ be a graph and let $M_s$ be an almost perfect matching in $G$.
An almost perfect matching $M_u$, if it exists, can be found in NC using Algorithm~\ref{algorithm-fix-up}.
\end{corollary}

This gives rise to the following algorithm that uses NC procedure for finding $O(\sqrt{n})$-planar separators~\cite{MILLER86}.

\begin{algorithm}[H]
\begin{algorithmic}[1]
\If{$G$ is a single vertex}
\State{Return $G$}
\EndIf
\State{Split the graph into two parts $G_1$ and $G_2$ using planar separator $T$.}
\State{Recurse on $G_1\cup T$ and $G_2$ to find maximum size matching $M_1$ and $M_2$.}
\State{Let $F$ be the set of free vertices.}
\State{For each $v\in F$ add a new node $v'$.}
\State{Add edge $vv'$ to $E$ and $M$.}
\ForAll{$v\in F$}\Comment{Sequentially}
\State{Remove $v'$ and $vv'$ from $G$.}
\If{there exists an alternating path $p$ from $v$ to any vertex in $F$}
\State{apply $p$ to $M$}
\Else
\State{Readd $v'$ to $V$ and $vv'$ to $E$ and $M$.}
\EndIf
\EndFor
\State{Remove $v'$ for all $v\in F$.}
\State{Return $M$.}
\end{algorithmic}
\caption{Finds a maximum cardinality matching in graph $G$.}
\end{algorithm}

After going up from the recursion on $G_1$ and $G_2$ there might be at most $O(\sqrt{n})$ augmenting paths with respect to $M_1\cup M_2$. Each
such path needs to be incident to a different vertex from the separator.

\begin{lemma}
A maximum size matching in planar non-bipartite graph can be computed in $O(\sqrt{n}\log n)$ time on PRAM.
\end{lemma}

\section{Min-Cost Flow}
\label{appendix-min-cost-flow}
This section presents an algorithm for computing min-cost planar flow. We are given a directed planar
network $N=(V,E)$. The edges have integral capacities given by $c:E\to [1, \poly(n)]$ and integral costs $a:E\to [0, \poly(n)]$.
Moreover, each vertex has integral demand $b:V\to [-\poly(n), \poly(n)]$.

We construct a bipartite graph $G_N$ whose maximum $f$-factor has weight equal to the min-cost of flow in $N$.
The following vertex-splitting construction was first given by \cite{gabow-tarjan-89} (see Figure~\ref{fig:f-factors} c)):
\begin{itemize}
\item for each $v\in V$ place vertices $v_{in}, \, v_{out}$ in $G_N$;
\item for each $v\in V$ add $c(\delta(v))$ copies of edge $v_{in}v_{out}$ to $G_N$;
\item for each $(u,v)\in E$ add  $c(u,v)$ copies of edge $u_{out}v_{in}$ to $G_N$;
\item for $v\in V$, if $b(v)>0$ set $f(v_{in})=c(\delta(v))$ and $f(v_{out})=c(\delta(v))+b(v)$
otherwise set $f(v_{in})=c(\delta(v))-b(v)$ and $f(v_{out})=c(\delta(v))$;
\item for each copy of the edge $u_{out}v_{in}$ set $w(u_{out},v_{in})=a(u,v)$.
\end{itemize}

\begin{corollary}[\cite{gabow-tarjan-89}]
\label{corollary-max-flow}
Let $N$ be the flow network. The weight of the minimum $f$-factor in $G_N$ is equal to the minimum cost of a flow in $N$.
\end{corollary}

Observe that for $G_N$ we have $f(V)\le 4c(E)+b(V)$, i.e., $f(V)$ is polynomial.
We note, that the reduction does not preserve planarity, e.g., consider a vertex od degree $4$ with in, out, in, out edges
in circular order. We can overcome this problem by first reducing the degree of the graph to $3$. This can be done
by replacing a vertex $v$ of degree $d$ by a directed cycle $C_v$ of length $d$ and connecting each edge incident to $v$ to
separate vertex on the cycle -- see Figure~\ref{fig:f-factors}~b). The capacity of the edges on $C_v$ should be $c(\delta(v))$ whereas their cost is $0$. The demand $b(v)$ is
assigned to one of the vertices of the cycle, whereas the other vertices have $0$ demand.
We note that degree $3$ vertices can have only two possible configuration of incident edges, i.e., in, in, out, and in, out, out.
Now in both these cases the above vertex splitting preserves planarity.

\begin{theorem}
Let $N=(V,E)$ be a planar flow network with integral
capacities $c:E\to [1, \poly(n)]$, integral costs $a:E\to [0, \poly(n)]$ and integral demands $b:V\to [-\poly(n), \poly(n)]$.
The minimum cost flow in $N$ can be computed in NC.
\end{theorem}

Given two vertices $s,t\in V$ the min-cost $st$-flow problem of value $f^*$ can be easily computed using the
above theorem. And if we want to find maximum flow of minimum cost, we first need to know the max-flow value
that can be either computed using~\cite{miller-95} or using the binary search.

\section{Maximum Multiple-Source Multiple-Sink Flow}
In this section we modify the idea from the previous section to handle the case when source and sink demands are not fixed but need to be maximized.
We are given a directed planar network $N=(V,E)$ with integral edge capacities given by $c:E\to [1, \poly(n)]$.
Moreover, each vertex has integral demand $b:V\to [-\poly(n), \poly(n)]$. Vertices $v\in V$ such that $b(v)\ge0$ are called
{\em sources} and we require for them $0\le f(v) \le b(v)$. Vertices $v\in V$ such that $b(v)\le 0$ are called {\em sinks}
and we require $b(v)\le f(v) \le 0$. For remaining vertices $v$, i.e., when $b(V)=0$ we need to have $f(v)=0$.
The {\em maximum multiple-source multiple-sink flow} in $N$ is the flow that maximizes value $f(S)$, where $S$ is the set of all sources.

In order to solve this problem we need to combine ideas from the previous section with the ideas used for computing maximum size bipartite matching, i.e.,
we first multiply all the demands by $2$ and then introduce self loops on source and sink vertices.
We construct a bipartite graph $G_N$ where maximum $f$-factor has weight equal to the maximum multiple-source multiple-sink flow in $N$.
\begin{itemize}
\item for each $v\in V$ place vertices $v_{in}, \, v_{out}$ in $G_N$;
\item for each $v\in V$ add $2c(\delta(v))$ copies of edge $v_{in}v_{out}$ to $G_N$;
\item for each $(u,v)\in E$ add  $2c(u,v)$ copies of edge $u_{out}v_{in}$ to $G_N$;
\item for each $v\in V$, if $b(v)\ge 0$ set $f(v_{in})=2c(\delta(v))$, $f(v_{out})=2c(\delta(v))+2b(v)$
and add $b(v)$ self loops to $v_{out}$;
\item for each $v\in V$. if $b(v)<0$ set $f(v_{in})=2c(\delta(v))-2b(v)$, $f(v_{out})=2c(\delta(v))$
and add $b(V)$ self loops to $v_{\in}$;
\item for each copy of the edge $u_{out}v_{in}$ set $w(u_{out},v_{in})=a(u,v)$.
\item set weight $1$ to all edges in $\delta(v_{out})$ for all $v\in S$ and weight $0$
to all other edges.
\end{itemize}

The next corollary is an immediate consequence of the above construction and the integrality of the flow problem.
\begin{corollary}
\label{corollary-max-flow-2}
Let $N$ be the flow network. The weight of the maximum $f$-factor in $G_N$ is equal to the maximum multiple-source multiple-sink flow in $N$.
\end{corollary}

Observe that for $G_N$ we have $f(V)\le 8c(E)+2b(V)$. Again, we note, that the reduction does not preserve planarity. Hence,
we first need to reduce the degree of the graph to $3$ as described in the previous section.

\begin{theorem}
Let $N=(V,E)$ be a planar flow network with integral
capacities $c:E\to [1, \poly(n)]$, integral costs $a:E\to [0, \poly(n)]$ and integral demands $b:V\to [-\poly(n), \poly(n)]$.
The minimum cost flow in $N$ can be computed in NC.
\end{theorem}

\newpage
\begin{figure}[p]
\centering
    \includegraphics[width=0.7\textwidth]{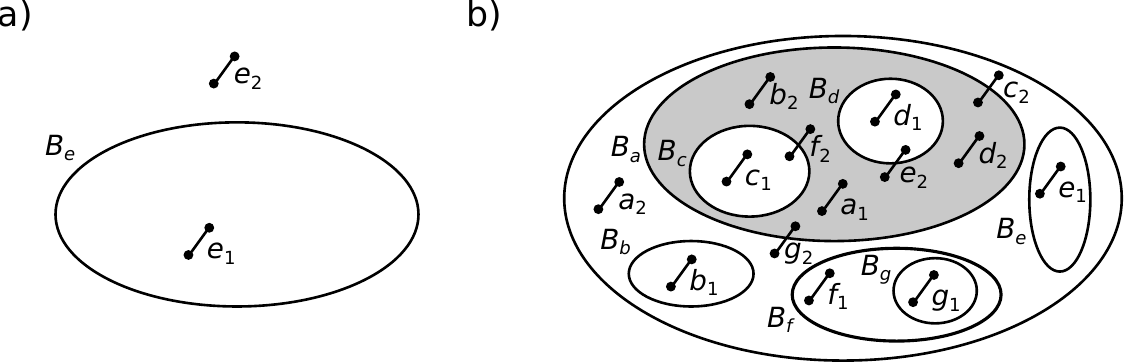}
    \caption{For each blossom $B_e$ there needs to be a pair of edges: $e_1$ inside and edge $e_2$ outside of $B_e$ -- Figure a).
    Figure b): when we recurse from Algorithm~\ref{algorithm-matching-recurse} to Algorithm~\ref{algorithm-degree-reduction} the blossom tree is empty, and we
    are recursing on one of the regions into which blossoms divide the plane -- marked with gray. Observe that at least one edge
    from each pair is not incident to this region.
    \label{fig:blossoms}}
\end{figure}

\begin{figure}[p]
\centering
    \includegraphics[width=0.5\textwidth]{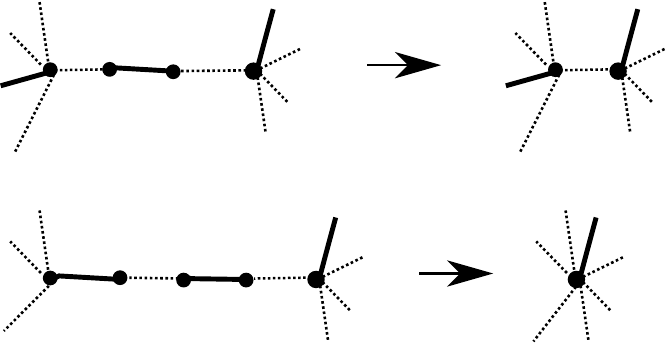}
    \caption{
    Figure a) shows manipulation of a path composed out of degree $2$ vertices by Algorithm~\ref{algorithm-degree-reduction}.
    Afterwards, degree $2$ vertices are independent. Vertices of degree $2$ are marked white and matched edges are drawn with solid lines.
    \label{fig:degree-2}}
\end{figure}

\begin{figure}[p]
\centering
    \includegraphics[width=0.99\textwidth]{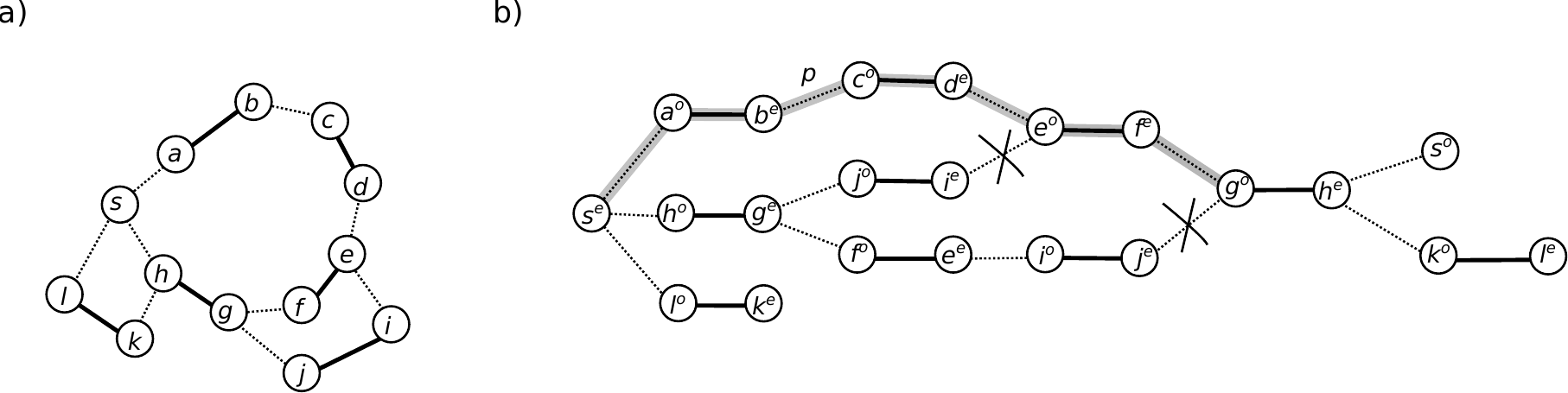}
    \caption{Figure a) shows a factor-critical graph, matched edges are marked with solid lines and $s$ is the free vertex.
    Figure b) presents layered graph $G_L$ from Algorithm~\ref{algorithm-fix-up}. There are two non-simple paths represented
    by this graph that go from $g^e$ to $g^o$. These paths are destroyed by the algorithm as $g^o$ is reachable via
    a simple path $p$ marked with grey.
    \label{fig:single-path}}
\end{figure}
\begin{figure}[p]
\centering
    \includegraphics[width=0.99\textwidth]{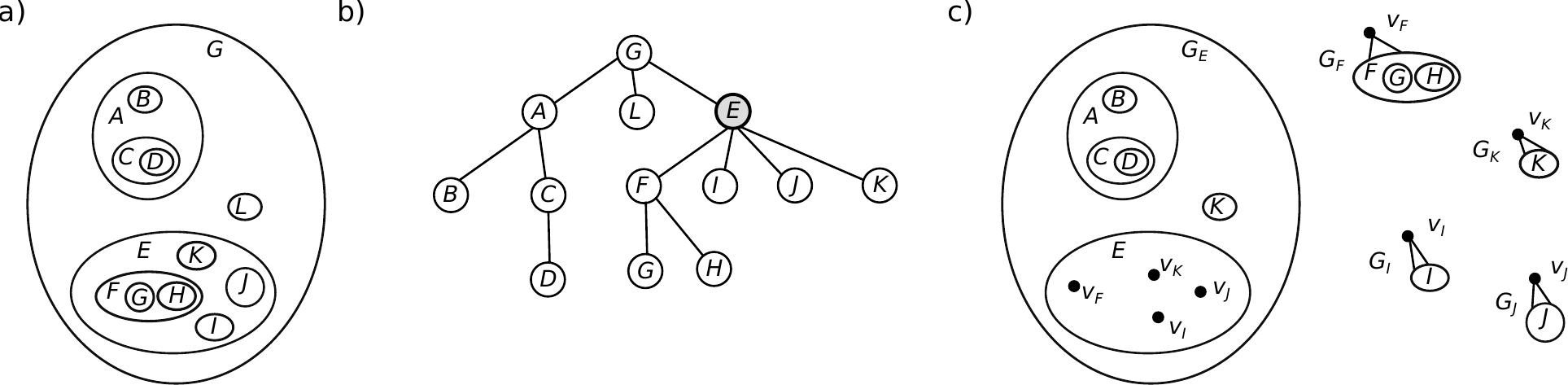}
    \caption{Figure a) shows a laminar family of sets. Figure b) gives the same family using laminar tree. Vertex separator -- node $E$ is marked with grey.
     Figure c) shows the graphs on which we recurse in Algorithm~\ref{algorithm-matching-recurse} using the marked separator $E$.\label{fig:recurse-tree}}
\end{figure}
\begin{figure}[p]
\centering
    \includegraphics[width=0.99\textwidth]{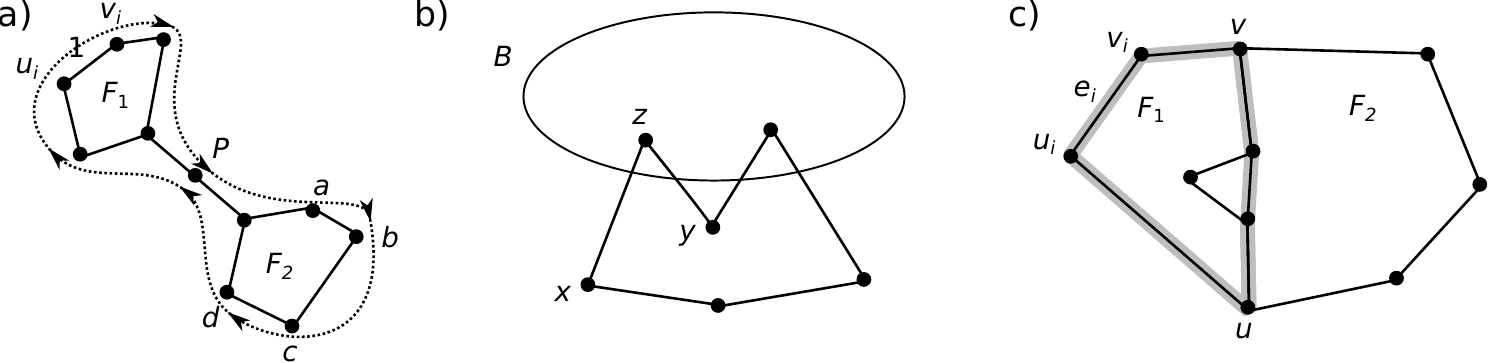}
    \caption{Figure a) shows a non-simple even length cycle composed out of two faces $F_1, F_2$ and a path $P$. We have $\pi_a+\pi_b=0$, $\pi_b+\pi_c=0$ and $\pi_c+\pi_d=0$, what
    gives $\pi_a-\pi_c=0$ and $\pi_a+\pi_d=0$. In Figure b) even length cycle does not contain an edge inside a blossom. In such a case $\pi_z+\pi_x+\pi_B=0$ and
    $\pi_z+\pi_y+\pi_B=0$, so $\pi_x-\pi_y=0$. Figure c) shows a case when two faces $F_1$ and $F_2$ share some edges. In such a case there always exists an even length cycle -- marked with gray.
    \label{fig:double-faces}}
\end{figure}

\begin{figure}[p]
\centering
    \includegraphics[width=0.99\textwidth]{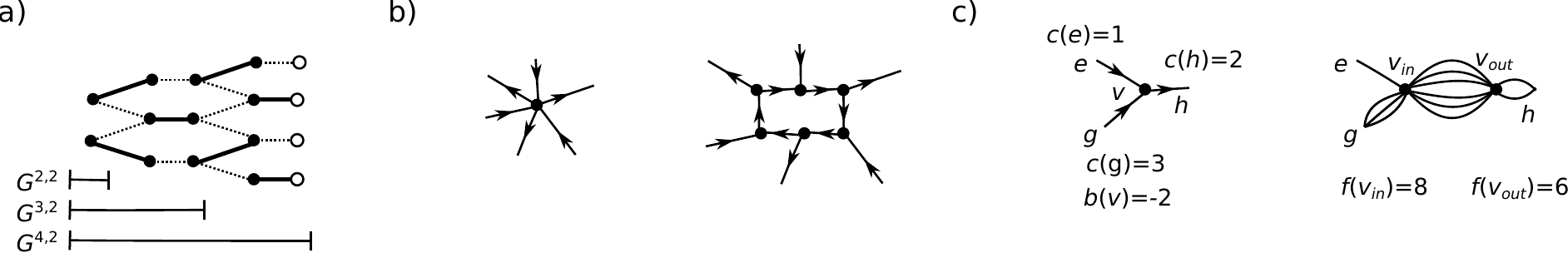}
    \caption{Figure a) shows the inductive construction of gadget $G^{4,2}$. The interface vertices are marked with white.
    Observe that exactly two of them can be made free. Figure b) presents a replacement of a vertex with a directed cycle.
    Figure c) presents the standard reduction of the flow problem to $f$-factor problem, when the degree of the vertex is $3$. The
    resulting graph is planar.
    \label{fig:f-factors}}
\end{figure}

\end{document}